\documentclass[twocolumn]{autart} 
\usepackage{url}
\usepackage{amsmath,amssymb,amsfonts}
\usepackage{algorithmic}
\usepackage{graphicx}
\usepackage{textcomp}
\usepackage{subcaption}
\usepackage{booktabs}
\usepackage{multirow}
\newenvironment{proof}{\noindent\textit{Proof.}}{\hfill$\square$\par}
\newtheorem{definition}{\textbf{Definition}}

\newtheorem{Lemma}{\bf{Lemma}}
\newtheorem{Assumption}{\bf{Assumption}}
\newtheorem{Theorem}{\bf{Theorem}}

\newtheorem{remark}{\bf{Remark}}

\begin{document}
\sloppy
\begin{frontmatter}

\title{Distributionally Robust Data-Driven Predictive Control for Stochastic LTI Systems}

\thanks[funding]{Corresponding author: Mirhan Urkmez.}

\author[aau]{Mirhan Urkmez}\ead{mu@es.aau.dk},    
\author[aau]{Shahab Heshmati-Alamdari}\ead{shhe@es.aau.dk},               

\address[aau]{ Section of Automation \& Control, Department of Electronic Systems, Aalborg University}  


\begin{abstract}
We propose a distributionally robust data-driven predictive control framework for stochastic linear time-invariant systems with unknown dynamics and disturbance distributions. We use an offline trajectory to fit the subspace predictive control (SPC) predictor via least squares and construct an empirical distribution of the prediction residuals as a proxy for the unknown disturbance distribution. We then center a Wasserstein ambiguity set around this estimate and minimize the worst-case expected cost while enforcing probabilistic output constraint satisfaction over all distributions in the set. The resulting problem admits a tractable reformulation with an equivalent direct data-driven form, eliminating the need for explicit predictor identification. Using finite-sample concentration results, we provide a data-driven Wasserstein radius such that, with high probability, the true expected cost is bounded above by the tractable objective and output constraints are satisfied with respect to the true disturbance distribution. Numerical simulations validate the framework against existing methods under various disturbance conditions and cost functions.
\end{abstract}

\end{frontmatter}

\section{Introduction}

The increasing availability of data from dynamical systems has driven growing interest in Data-Driven Predictive Control (DDPC) methods. Historically, data has been used to identify a system model, which is then used to design a predictive control law. These approaches are called indirect DDPC methods \cite{9705109}, and despite their success across a wide range of industrial applications, they require an accurate mathematical model of the system, which can be difficult and costly to obtain in practice. Recently, direct DDPC methods have attracted significant interest \cite{9705109}, where model identification is bypassed entirely and data is used directly to predict future trajectories and compute control actions.

A large body of direct DDPC methods has been developed building on Willems' fundamental lemma \cite{willems2005note}, which establishes that under a persistence of excitation condition, the entire behavior of a deterministic LTI system can be parametrized via a Hankel matrix of offline data through a coefficient vector. Following this,a direct DDPC framework known as Data-Enabled Predictive Control (DeePC) was proposed in \cite{8795639} 
with an equivalence to Model Predictive Control (MPC) for deterministic 
LTI systems. However, when the system is subject to stochastic disturbances, the Hankel-based parametrization is no longer exact and predictions can deviate significantly from the true system behavior.

To address this, regularization techniques have been proposed. In \cite{8795639}, $\ell_1$
 regularization is proposed to enforce consistency with the underlying system. Several equivalence results between direct DDPC and Subspace Predictive Control (SPC) \cite{FAVOREEL19994004} have guided the design of regularization terms. SPC constructs a multi-step predictor by fitting a parametric model to offline trajectory data via least squares. Since the identification process reduces the effect of noise compared to using raw data , direct DDPC methods that establish an equivalence with SPC become less sensitive to noise. In \cite{9705109}, a connection with the SPC predictor is established and a corresponding regularization is introduced, and \cite{10896586} derives an analogous equivalence with the causal SPC predictor. 

Another line of work provides formal guarantees for direct DDPC under explicit disturbance assumptions, including bounded noise \cite{9109670,berberichConst} and stochastic formulations \cite{9999463,11192654}, the latter requiring exact knowledge of the disturbance distribution. A distributionally robust formulation for nonlinear systems is proposed in \cite{9488221}, which does not require this knowledge though it requires multiple offline trajectories with identical initial conditions and inputs.

In this work, we propose a distributionally robust direct DDPC formulation for LTI systems subject to stochastic disturbances with unknown probability distributions, obtained via an equivalent SPC formulation. The main contributions are as follows. (i) We propose a distributionally robust SPC formulation with a Wasserstein ambiguity set over future disturbance distributions, constructed around an empirical distribution derived from the prediction residuals of an SPC predictor estimated via least squares from an offline trajectory. The formulation minimizes the worst-case expected value of a function of the actual system output while enforcing probabilistic output constraint satisfaction for all distributions in the ambiguity set. (ii) We derive a tractable reformulation and show that it admits an equivalent direct DDPC form. (iii) Using finite-sample concentration results, we provide a data-driven choice of the ambiguity set radius that guarantees the true disturbance distribution is contained within the ambiguity set with high probability.  This in turn yields a probabilistic upper bound on the expected cost and a constraint satisfaction guarantee, both holding with high probability.


\subsection{Notation}
 For a sequence $\{z_k\}_{k=1}^T$ with $z_k\in\mathbb{R}^\eta$, we write $z_{[k,j]}=\begin{bmatrix}z_k^\top&\cdots&z_j^\top\end{bmatrix}^\top$ and $z:=z_{[1,T]}$. The Hankel matrix of depth $L$ is $H_L(z)=\begin{bmatrix}z_{[1,L]}&\cdots&z_{[T-L+1,T]}\end{bmatrix}\in\mathbb{R}^{\eta L\times(T-L+1)}$. We write $(\cdot)_+:=\max\{\cdot,0\}$ and $\|\cdot\|_r$ for the $\ell_r$-norm.

\section{Problem Formulation}
\label{sec:problem}

We consider the stochastic discrete-time LTI system 
\begin{equation}
\begin{aligned}
x_{k+1} &= A x_k + B u_k + w_k, \\
y_k     &= C x_k + D u_k + v_k,
\end{aligned}
\label{eq:lti}
\end{equation}
where $x_k\in \mathbb{R}^n$, $u_k\in \mathbb{R}^m$, $y_k\in \mathbb{R}^p$ are 
states, inputs, and outputs of the system, respectively. The terms $w_k\in \mathbb{R}^n$ 
and $v_k \in \mathbb{R}^p$ represent process and measurement noise, drawn from unknown 
probability distributions $\mathbb{P}_w$ and $\mathbb{P}_v$ over support sets 
$W \subseteq \mathbb{R}^n$ and $V \subseteq \mathbb{R}^p$. We assume that $\{w_k\}$ 
and $\{v_k\}$ are i.i.d.\ sequences over time. The system matrices 
$A \in \mathbb{R}^{n \times n}$, $B \in \mathbb{R}^{n \times m}$, 
$C \in \mathbb{R}^{p \times n}$, $D \in \mathbb{R}^{p \times m}$ are unknown.

We aim to develop a receding horizon predictive controller for~\eqref{eq:lti}, where 
at each time step $k$, given the most recent input-output measurements
\begin{equation}
\label{eq:pfnotation}
u_p := u_{[k-T_p,\,k-1]}, \qquad y_p := y_{[k-T_p,\,k-1]},
\end{equation}
of length $T_p$, an optimization problem is solved over the future input 
sequence $u_f := u_{[k,\,k+T_f-1]}$ with prediction horizon $T_f$ minimizing a cost function in terms of the inputs $u_f$ and future  outputs $y_f := y_{[k,\,k+T_f-1]}$. The 
first element of the optimal input sequence is applied to the system, and 
the past window is updated at the next step.  To express $y_f$ in terms of $u_p$, $y_p$, $u_f$, and the disturbances, 
we make the following assumption.
\begin{Assumption}\label{as:obs}
\!\!The pair $(A, C)$ is observable and $T_p \!\ge\! n$.
\end{Assumption}
Under Assumption~1, the future output $y_f$ admits the multi-step predictor 
representation
\begin{equation}
\label{eq:spc}
y_f = K m_f + \xi_f,
\end{equation}
where $m_f = \begin{bmatrix} u_p^\top & y_p^\top & u_f^\top \end{bmatrix}^\top$. 
The exact expressions for $K$ and $\xi_f$ are obtained by eliminating the 
initial state via the observability of $(A,C)$ and are omitted for brevity. 
The matrix $K$ is the unknown multi-step predictor, which is a function of 
the system matrices $(A,B,C,D)$. The term $\xi_f \in \mathbb{R}^{pT_f}$ is 
the multi-step disturbance, which is a function of the past disturbances 
$(w_p, v_p)$ and future disturbances $(w_f, v_f)$, where these vectors are 
defined analogously to~\eqref{eq:pfnotation}. The distribution of $\xi_f$ 
is denoted by $\mathbb{P}_{\xi_f}$ and supported on $\Xi \subseteq \mathbb{R}^{pT_f}$. 
We also define $\mathbb{P}_{y_f|m_f}$ as the probability distribution of $y_f$ 
given a fixed $m_f$, supported on $\mathcal{Y} \subseteq \mathbb{R}^{pT_f}$.

Since the system matrices are unknown, we assume the availability of an 
offline input-output trajectory $\{u_k^d,y_k^d\}_{k=1}^{T}$ of length $T$, 
collected prior to the operation of the controller. Setting $L := T_p + T_f$, the Hankel matrices 
$H_L(u^d)$ and $H_L(y^d)$ are partitioned into past and future blocks as
\begin{equation}
\label{eq:hankels}
H_L(u^d)=
\begin{bmatrix}
U_p \\ U_f
\end{bmatrix}, \qquad
H_L(y^d)=
\begin{bmatrix}
Y_p \\ Y_f
\end{bmatrix},
\end{equation}
where $U_p \in \mathbb{R}^{mT_p \times N}$, $U_f \in \mathbb{R}^{mT_f \times N}$ 
correspond to the first and last $T_p$, $T_f$ block rows of $H_L(u^d)$, 
$Y_p$, $Y_f$ are defined analogously, and $N = T - L + 1$. 

\subsection{Distributionally Robust Formulation}
\label{sec:dr}

Since $y_f$ is uncertain under stochastic disturbances, we consider the 
 expected cost $\mathbb{E}_{\mathbb{P}_{y_f|m_f}}\!\left[ f_1(u_f) + 
 f_2(y_f) \right]$ as our objective 
where $f_1$ and $f_2$ are the input and output cost functions, respectively. 
However, since both $K$ and $\mathbb{P}_{y_f|m_f}$ are unknown, 
$\mathbb{P}_{y_f|m_f}$ is inaccessible and this objective cannot be minimized 
directly. Moreover, any finite-sample estimate of $\mathbb{P}_{y_f|m_f}$ 
constructed from offline data will inevitably carry estimation error. To 
robustify against this uncertainty, we adopt a distributionally robust 
formulation in which we construct an ambiguity set $\mathcal{P}$ of candidate 
distributions centered around a data-driven empirical estimate of 
$\mathbb{P}_{y_f|m_f}$, and minimize the worst-case expected cost over 
$\mathcal{P}$.

We consider output constraints of the form $h(y_f) \le 0$. Since $y_f$ is stochastic, this constraint cannot be 
enforced deterministically. Instead, we require it to hold in a probabilistic 
sense via the Conditional Value-at-Risk (CVaR), defined below, which serves 
as an upper bound on the Value-at-Risk and therefore guarantees satisfaction 
of the corresponding chance constraint~\cite{nemirovski}. We require 
$\mathrm{CVaR}_{1-\beta}^{Q}\!\left( h(y_f) \right) \leq 0$ to hold for all 
$Q \in \mathcal{P}$.

\begin{definition}[Conditional Value-at-Risk]
Let $\omega \in \Omega \subseteq \mathbb{R}^r$ be a random variable with 
 distribution $\mathbb{P}_{\omega}$. For a function 
$\phi : \mathbb{R}^r \to \mathbb{R}$, the Conditional Value-at-Risk (CVaR) 
at confidence level $1-\beta \in (0,1)$ is defined as
\begin{equation}
\mathrm{CVaR}^{1-\beta}_{\mathbb{P}_{\omega}}\big(\phi(\omega)\big)\! := 
\inf_{t \in \mathbb{R}} \left\{ \frac{1}{\beta}\, \mathbb{E}_{\mathbb{P}_{\omega}} 
\big[ (\phi(\omega)\! - t)_+ \big] \!+\! t \right\},
\end{equation}
where $(\cdot)_+ := \max\{ \cdot, 0 \}$.
\end{definition}

The resulting distributionally robust optimization problem is formulated as
\begin{subequations}
\label{eq:dr_general}
    \begin{align}
    \min_{u_f} \; & f_1(u_f) + \sup_{Q \in \mathcal{P}} \mathbb{E}_{Q}\!\left[ f_2(y_f) \right]  \\
    \text{s.t.} \; & u_f \in \mathcal{U}, \\
    & \sup_{Q \in \mathcal{P}} \mathrm{CVaR}^{1-\beta}_{Q}\!\left( h(y_f) \right) \leq 0,
    \end{align}
\end{subequations}
where $\mathcal{U} \subseteq \mathbb{R}^{mT_f}$ is the input constraint set. 
\subsection{Ambiguity Set Construction}
\label{sec:ambiguity}
The construction of $\mathcal{P}$ proceeds in two stages. First, we use offline data to compute an empirical estimate of the  distribution  $\mathbb{P}_{y_f|m_f}$.
Second, we define a neighborhood of distributions around this estimate to account for estimation errors. Define the 
offline regressor matrix
\begin{equation}
M := \begin{bmatrix} U_p \\ Y_p \\ U_f \end{bmatrix} \in \mathbb{R}^{(m(T_p+T_f) + pT_p) \times N},
\end{equation}
where $N = T - L + 1$ and $L = T_p + T_f$. The SPC predictor is estimated by 
least squares:
\begin{equation}
\hat{K} := \underset{\tilde{K}}{\arg\min} \left\| Y_f - \tilde{K} M \right\|_F^2 
= Y_f M^{\dagger},
\label{eq:ls_khat}
\end{equation}
where $M^{\dagger}$ is the Moore-Penrose pseudoinverse of $M$. Using this estimate, we generate a set of approximate residuals $\hat{\xi}_f^{(i)}$ using offline data as
\begin{equation}
\hat{\xi}_f^{(i)} := \left(Y_f - \hat{K} M\right) e_i, \qquad i = 1, \dots, N,
\label{eq:empirical_residuals}
\end{equation}
where $e_i \in \mathbb{R}^N$ is the $i$-th standard basis vector. For a given 
$m_f$, we form predictions of $y_f$ by shifting the SPC prediction 
$\hat{K}m_f$ with the approximate residual samples
\begin{equation}\
\hat{y}_f^{(i)}(m_f) := \hat{K} m_f + \hat{\xi}_f^{(i)}, \qquad i = 1, \dots, N.
\label{eq:empirical_yf}
\end{equation}
The resulting empirical distribution
\begin{equation}
\hat{\mathbb{P}}_{y_f \mid m_f} := \frac{1}{N} \sum_{i=1}^{N} \delta_{\hat{y}_f^{(i)}(m_f)}
\label{eq:empirical_distribution}
\end{equation}
serves as our prediction distribution for $y_f$ given $m_f$, around which we 
construct the ambiguity set $\mathcal{P}$. For the finite-sample 
analysis in Section~\ref{sec:guarantees}, we  define 
\begin{equation}
\bar{\mathbb{P}}_{y_f \mid m_f} := \frac{1}{N} \sum_{i=1}^{N} \delta_{\bar{y}_f^{(i)}(m_f)},
\qquad \bar{y}_f^{(i)}(m_f) := K m_f + \bar{\xi}_f^{(i)},
\label{eq:oracle_distribution}
\end{equation}
where $\bar{\xi}_f^{(i)} = \left(Y_f - K M\right) e_i$. This is unknown since 
$K$ is unknown, but serves as an intermediate reference in the finite-sample 
analysis. We construct $\mathcal{P}$ as a Wasserstein ball centered at $\hat{\mathbb{P}}{y_f \mid m_f}$, as the Wasserstein metric enables a tractable reformulation of the optimization problem \eqref{eq:dr_general}.
\begin{definition}
Let $r \in \left[1, \infty \right]$ and $\mathcal{M}(\mathcal{Y})$ be the set of all probability measures $Q$ 
supported on $\mathcal{Y} \subseteq \mathbb{R}^{pT_f}$ with 
$\mathbb{E}_Q[\|y\|_r] < \infty$. The r-Wasserstein metric 
$d_{W_r} : \mathcal{M}(\mathcal{Y}) \times \mathcal{M}(\mathcal{Y}) \to \mathbb{R}_{\ge 0}$ 
is defined as
\begin{equation}
d_{W_r}(Q_1, Q_2) := \bigg(\inf_{\Pi} \left\{ \int_{\mathcal{Y}^2} \|y_1 - y_2\|_r^r \, 
\Pi(\mathrm{d}y_1, \mathrm{d}y_2)\right\}\bigg)^{1/r},
\label{eq:wasserstein}
\end{equation}
where $\Pi$ ranges over all joint measures on $\mathcal{Y} \times \mathcal{Y}$ 
with marginals $Q_1$ and $Q_2$.
\end{definition}

The ambiguity set is defined as the Wasserstein ball of radius $\varepsilon > 0$ 
centered at $\hat{\mathbb{P}}_{y_f \mid m_f}$ for some $r \in \left[1, \infty \right]$
\begin{equation}
\mathcal{P} \!:=\! \mathcal{B}_\varepsilon\!\left(\hat{\mathbb{P}}_{y_f \mid m_f}\right) 
\!:=\! \Big\{ Q \in \mathcal{M}(\mathcal{Y}) \;\big|\; 
d_{W_r}\!\left(Q,\, \hat{\mathbb{P}}_{y_f \mid m_f}\right) \le \varepsilon \Big\}.
\label{eq:wasserstein_ball}
\end{equation}
Substituting into~\eqref{eq:dr_general} gives the distributionally robust SPC problem
\begin{subequations}
\label{eq:dr_wass}
    \begin{align}
    \min_{u_f} \; & f_1(u_f) + \sup_{Q \in \mathcal{B}_\varepsilon(\hat{\mathbb{P}}_{y_f \mid m_f})} 
    \mathbb{E}_{Q}\!\left[ f_2(y_f) \right]  \\
    \text{s.t.} \; & u_f \in \mathcal{U}, \\
    & \sup_{Q \in \mathcal{B}_\varepsilon(\hat{\mathbb{P}}_{y_f \mid m_f})} 
    \mathrm{CVaR}^{1-\beta}_{Q}\!\left( h(y_f) \right) \leq 0.
    \end{align}
\end{subequations}
\section{Main Results}
\subsection{Tractable Reformulation}
\label{sec:tract}
 We now provide a tractable reformulation of~\eqref{eq:dr_wass}.

\begin{Theorem}
\label{thm:tractable_spc}
Assume that $f_2$ and $h$ are convex and Lipschitz continuous with constants $L_{\mathrm{obj}} > 0$ and $L_{\mathrm{con}} > 0$ with respect to the $r$-norm, respectively. Then, the optimal value of the distributionally robust SPC problem~\eqref{eq:dr_wass} is upper bounded by the optimal value of
\begin{subequations}\label{eq:tractable_spc}
\begin{align}
&\min_{u_f,\,\tau,\,s_i} \quad  f_1(u_f) + \frac{1}{N}\sum_{i=1}^{N} f_2\!\left(  \hat{K} m_f + \hat{\xi}_f^{(i)}\right) + L_{\mathrm{obj}}\,\varepsilon \\
&\text{s.t.} \quad
 u_f \in \mathcal{U},\label{eq:tr_inp}\\
& -\tau\beta + L_{\mathrm{con}}\,\varepsilon + \frac{1}{N}\sum_{i=1}^{N} s_i \le 0, \label{eq:tr_cvar}\\
& \tau + h\!\left(  \hat{K} m_f + \hat{\xi}_f^{(i)}\right) \le s_i, \quad \forall i = 1,\dots,N, \label{eq:tr_slack}\\
& s_i \ge 0, \quad \forall i = 1,\dots,N. \label{eq:tr_slack_pos}
\end{align}
\end{subequations}
where $\hat{K}$ is given by \eqref{eq:ls_khat}. 
Furthermore, any feasible solution $u_f$ to~\eqref{eq:tractable_spc} satisfies
\begin{equation}
    \mathrm{CVaR}_{Q}^{1-\beta}\!\left( h(y_f) \right) \le 0, \quad 
    \forall Q \in \mathcal{B}_\varepsilon\!\left(\hat{\mathbb{P}}_{y_f \mid m_f}
    \right).
\end{equation}
\end{Theorem}
\begin{proof}
Proof is given in the Appendix~\ref{app:th2}.
\end{proof}
We now show that the tractable SPC problem~\eqref{eq:tractable_spc}, which requires knowledge of $\hat{K}$, admits an 
equivalent reformulation solely in terms of the offline Hankel matrices, 
without explicit identification of $\hat{K}$.

\begin{Theorem}
\label{thm:tractable_deepc}
Problem~\eqref{eq:tractable_spc} is equivalent to
\begin{subequations}\label{eq:tractable_deepc}
\begin{align}
\min_{g,\,\tau,\,s_i} \quad &\! f_1(U_f g)\! + \!\frac{1}{N}\sum_{i=1}^{N} f_2\!\left(Y_f g + Y_f(I-P)e_i\right) + L_{\mathrm{obj}}\,\varepsilon\\
\text{s.t.} \quad
& U_p g = u_p,\\
& Y_p g = y_p,\\
& (I-P)g = 0,\\
& U_f g \in \mathcal{U},\\
& -\tau\beta + L_{\mathrm{con}}\,\varepsilon + \frac{1}{N}\sum_{i=1}^{N} s_i \le 0,\\
& \tau + h\!\left(Y_f g + Y_f(I-P)e_i\right) \le s_i, \quad \forall i = 1,\dots,N,\\
& s_i \ge 0, \quad \forall i = 1,\dots,N,
\end{align}
\end{subequations}
where $P := M^{\dagger} M$ is the orthogonal projector onto the row space of $M$.
\end{Theorem}
\begin{proof}
It was shown in \cite{9705109} that $\hat{K} m_f$ can equivalently be written as $Y_f g$ if and only if $g$ satisfies
\begin{equation*}
    U_p g = u_p, \quad Y_p g = y_p, \quad U_fg=u_f, \quad (I-P)g = 0.
\end{equation*}
Furthermore, using $\hat{K} = Y_f M^{\dagger}$ and $P = M^{\dagger}M$, the 
empirical residuals \eqref{eq:empirical_residuals} satisfy
\begin{equation*}
    \hat{\xi}_f^{(i)} = \left(Y_f - \hat{K}M\right)e_i = 
    Y_f\left(I - M^{\dagger}M\right)e_i = Y_f(I-P)e_i.
\end{equation*}
Consequently, the empirical output samples \eqref{eq:empirical_yf} become
\begin{equation*}
    \hat{y}_f^{(i)} = \hat{K} m_f + \hat{\xi}_f^{(i)} = 
    Y_f g + Y_f(I-P)e_i, \quad i=1,\dots,N,
\end{equation*}
where the second equality holds whenever $g$ satisfies the above constraints. 
Substituting these into~\eqref{eq:tractable_spc} and setting $u_f = U_f g$ 
yields~\eqref{eq:tractable_deepc}.
\end{proof}

\begin{remark}
The equivalence in Theorem~\ref{thm:tractable_deepc} relies on \cite{9705109}. The causal SPC predictor equivalence of \cite{10896586} can also be used, where the empirical distribution is constructed from the causal predictor residuals and an equivalent direct DDPC form is obtained through that equivalence result.
\end{remark}

\subsection{Finite-Sample Guarantee}
\label{sec:guarantees}
We now establish finite-sample probabilistic guarantees for the proposed formulation by characterizing a radius $\varepsilon(\alpha,m_f)$ 
such that  $\mathbb{P}_{y_f|m_f} \in \mathcal{B}_{\varepsilon(\alpha,m_f)}(\hat{\mathbb{P}}_{y_f|m_f})$ with 
probability at least $1-\alpha$, which guarantees the cost bound 
and CVaR constraint satisfaction of Theorem~\ref{thm:tractable_spc} with 
the same probability. To this end, Let $\mathbb{P}$ denote the joint measure of the initial state and noise 
sequences governing the offline data-generating process.  We rely on the following assumption bounding the mismatch
between estimated predictor $\hat{K}$ and the true predictor $K$. 

\begin{Assumption}\label{as:SPC_estimator}
For a given confidence level $\alpha \in (0,1)$, let $\gamma(\alpha) > 0$ be 
such that
\begin{equation}
\mathbb{P}\Big( \|K - \hat{K}\|_r \le \gamma(\alpha) \Big) \ge 
1 - \alpha,
\end{equation}
\end{Assumption}
We refer to~\cite{9992469} and references therein for conditions under which this holds.

We will use the triangular inequality to find a bound on the Wasserstein distance between the true probability distribution $ \mathbb{P}_{y_f|m_f}$ and our empirical distribution $\hat{\mathbb{P}}_{y_f|m_f}$ as
\begin{align*}
&d_{W_r}\!\big(\hat{\mathbb{P}}_{y_f|m_f},\, \mathbb{P}_{y_f|m_f}\big) \le 
d_{W_r}\!\big(\hat{\mathbb{P}}_{y_f|m_f},\, \bar{\mathbb{P}}_{y_f|m_f}\big) \\
&+ d_{W_r}\!\big(\bar{\mathbb{P}}_{y_f|m_f},\, \mathbb{P}_{y_f|m_f}\big).
\end{align*}
The following two lemmas bound each term respectively.
\begin{Lemma}
\label{lem:wasserstein_estimation_error}
Under Assumptions~\ref{as:obs} and~\ref{as:SPC_estimator}, let $\alpha \in (0,1)$ specify a 
risk level. Then,
\begin{equation}
\mathbb{P}\left\{ d_{W_r}\big( \hat{\mathbb{P}}_{y_f|m_f},\, 
\bar{\mathbb{P}}_{y_f|m_f} \big) \ge \gamma\left(\tfrac{\alpha}{2}\right) 
\Psi_N^{1/r} \right\} \le \frac{\alpha}{2}
\end{equation}
where $\Psi_N = \frac{1}{N} \sum_{i=1}^{N} \|Me_i - m_f\|_r^r$.
\end{Lemma}
\begin{proof}
The proof follows from Lemma~5 in~\cite{9992469}, 
adapted to the $r$-Wasserstein distance  
instead of the $1$-Wasserstein. 
\end{proof}

We rely on the following result to bound
$d_{W_r}\!\big(\bar{\mathbb{P}}_{y_f|m_f},\, \mathbb{P}_{y_f|m_f}\big)$.

\begin{Lemma}\label{lem:wasserstein_concentration}
Assume that $\mathbb{E}_{\mathbb{P}_{y_f|m_f}}\!\left[\|y_f\|_r^q\right] < \infty$
for some  $q > r$. Then, under Assumption~\ref{as:obs},  there exists a constant $C>0$ such that
\begin{equation}
\mathbb{E}_{\mathbb{P}}\!\left[d_{W_r}^r\!\big(\bar{\mathbb{P}}_{y_f|m_f},\, \mathbb{P}_{y_f|m_f}\big)\right] \le \gamma(N),
\end{equation}
where
\begin{equation}
\gamma(N)=
C
\left\{
\begin{array}{@{}l@{\ }l@{}}
N^{-1/2}+N^{-(q-r)/q}, & \! r>d/2,\\[4pt]
N^{-1/2}\log(1+N)+N^{-(q-r)/q}, &\! r=d/2,\\[4pt]
N^{-r/d}+N^{-(q-r)/q}, &\! r\in(0,d/2),
\end{array}
\right.
\end{equation}
with $d=pT_f$. Consequently, for all $\kappa>0$,
\begin{equation}
\mathbb{P}\!\Big(d_{W_r}\!\big(\bar{\mathbb{P}}_{y_f|m_f},\, \mathbb{P}_{y_f|m_f}\big)\ge\kappa\Big)
\le
\frac{\gamma(N)}{\kappa^r}.
\end{equation}
\end{Lemma}
\begin{proof}
Proof is given in the Appendix~\ref{app:wass}.
\end{proof}

\begin{Theorem}
\label{thm:main_guarantee}
Under Assumptions~\ref{as:obs} and~\ref{as:SPC_estimator}, and the 
moment condition of Lemma~\ref{lem:wasserstein_concentration}, let 
$\alpha \in (0,1)$ be a risk level and define the ambiguity set radius
\begin{equation}
\varepsilon(\alpha, m_f) := \varepsilon_1(\alpha) \bigg(\frac{1}{N} \sum_{i=1}^{N} 
\|Me_i - m_f\|_r^r\bigg)^{1/r} + \varepsilon_2(\alpha),
\label{eq:epsilon_def}
\end{equation}
where $\varepsilon_1(\alpha) := \gamma\!\left(\tfrac{\alpha}{2}\right)$ and\
$\varepsilon_2(\alpha) :=
\left(\frac{2\,\gamma(N)}{\alpha}\right)^{1/r}$ where $\gamma(N)$
 is defined in Lemma~\ref{lem:wasserstein_concentration}. 
Let $g^\star$ be a feasible solution of problem~\eqref{eq:tractable_deepc} 
with radius $\varepsilon = \varepsilon(\alpha,m_f)$, and let $\hat{J}(g^\star)$ 
denote the objective value of~\eqref{eq:tractable_deepc} at $g^\star$, and 
define the true expected cost
\begin{equation}
J(g) := \mathbb{E}_{\mathbb{P}_{y_f|m_f}}\!\left[ f_1(U_f g) + f_2(y_f) 
\right].
\end{equation}
Then the following probabilistic guarantees hold with respect to 
$\mathbb{P}$:
\begin{equation}
\begin{aligned}
\mathbb{P}\!\left\{ J(g^\star) \le \hat{J}(g^\star) \right\} &\ge 1 - \alpha, \\
\mathbb{P}\!\left\{ \mathrm{CVaR}^{1-\beta}_{\mathbb{P}_{y_f|m_f}}\!\left( 
h(y_f) \right) \le 0 \right\} &\ge 1 - \alpha.
\end{aligned}
\end{equation}
\end{Theorem}
\begin{proof}
From Lemmas~\ref{lem:wasserstein_estimation_error} 
and~\ref{lem:wasserstein_concentration} and the union bound, it follows that
\begin{align*}
    &\mathbb{P} \big\{ d_{W_r}( \hat{\mathbb{P}}_{y_f|m_f}, \mathbb{P}_{y_f|m_f} ) \ge \varepsilon(\alpha) \big\} \\
    &\le \mathbb{P} \big\{ d_{W_r}( \hat{\mathbb{P}}_{y_f|m_f}, \bar{\mathbb{P}}_{y_f|m_f} ) \ge \varepsilon_1(\alpha) \Psi_N^{1/r} \big\} \\
    &\quad + \mathbb{P} \big\{ d_{W_r}( \bar{\mathbb{P}}_{y_f|m_f}, \mathbb{P}_{y_f|m_f} ) \ge \varepsilon_2(\alpha) \big\} \le \tfrac{\alpha}{2} + \tfrac{\alpha}{2} = \alpha,
\end{align*}
where $\Psi_N = \frac{1}{N} \sum_{i=1}^{N} \|Me_i - m_f\|_r^r$.
This means that with probability at least $1-\alpha$, the true distribution 
$\mathbb{P}_{y_f|m_f}$ lies inside the Wasserstein ball 
$\mathcal{B}_{\varepsilon(\alpha)}\!\big(\hat{\mathbb{P}}_{y_f|m_f}\big)$. 
The result then follows from Theorem~\ref{thm:tractable_spc}.
\end{proof}
\begin{remark}
\label{rem:epsilon_tuning}
The radius $\varepsilon(\alpha,m_f)$ in~\eqref{eq:epsilon_def} may be 
overly conservative in practice~\cite{Esfahani2015DatadrivenDR}. We 
therefore assign separate tuning parameters $\varepsilon_{\mathrm{obj}} = 
\varepsilon_1 \Psi_N^{1/r} + \varepsilon_2$ and a fixed scalar 
$\varepsilon_{\mathrm{con}}$ to the cost and constraint terms, respectively, 
and treat both as offline tuning parameters in Section~\ref{sec:sim}.
\end{remark}

\section{Numerical Example}
\label{sec:sim}
We illustrate the proposed method through simulation on an example system taken from \cite{BRESCHI2023110961}, which takes the form
\begin{equation}
\begin{aligned}
x(t+1) &= A x(t) + B u(t) + K e(t), \\
y(t) &= C x(t) + e(t),
\end{aligned}
\end{equation}
with matrices
\begin{align*}
&A = \begin{bmatrix} 0.7326 & -0.0861 \\ 0.1722 & 0.9909 \end{bmatrix}, \quad
B = \begin{bmatrix} 0.0609 \\ 0.0064 \end{bmatrix}, \\
&K = \begin{bmatrix} -0.5 \\ 0.5 \end{bmatrix}, \quad
C = \begin{bmatrix} 0 & 1.4142 \end{bmatrix}.
\end{align*}
The distribution of the innovation term $e(t)$ varies across experiments and is specified for each simulation scenario. The system is subject to  box constraints on both inputs and outputs
$y(t) \in [-2, 2]$, $ u(t) \in [-2, 2]$. 

We compare the proposed formulation against two existing methods implemented 
in a receding horizon fashion. The first is SPC, in which the output is 
predicted using $\hat{K}m_f$ with $\hat{K}$ estimated as in 
\eqref{eq:ls_khat}. The second is Reg-DeePC (Eq.~23 in~\cite{9705109}), which augments DeePC 
with an $\ell_1$ regularization term $\lambda_g\|g\|_1$, with weight 
selected via grid search and fixed across all experiments. We implement the 
proposed method in its direct DDPC form \eqref{eq:tractable_deepc}, 
referred to as DR-DDPC\footnote{Code available at \url{https://github.com/mirhanu/DR-DDPC}}. For all methods, output constraints, including the CVaR 
constraints in \eqref{eq:tractable_deepc}, are enforced softly by 
augmenting the cost with a weighted sum of squared violations.

The offline data  $\{u_k^d,y_k^d\}_{k=1}^{T}$ is generated from a single trajectory of length $T=200$, where the system is excited by a random control law $u(t) \sim \mathcal{N}(0, I_m)$. The Hankel matrices are constructed with a past horizon $T_p = 5$ and a prediction horizon $T_f = 10$, yielding $N = T - T_p - T_f + 1=186$ residual samples. Results are reported over 50 Monte Carlo simulations, regenerating offline data, disturbance realizations, and initial conditions in 
each iteration, with the same 50 realizations shared across all methods. For the distributionally robust formulations, we use Wasserstein ambiguity 
sets with an $\ell_2$-norm ($r=2$) and, following Remark~\ref{rem:epsilon_tuning}, 
set $\varepsilon_1 = \varepsilon_2 = 10^{-3}$ for the objective term. For computational efficiency, the full set 
of $N$ residuals is used in the cost term, while only $20$ residuals are 
used in the constraint term. Both $\varepsilon_{\mathrm{con}}$ and $\beta$ are swept in the first 
experiment and fixed to $\varepsilon_{\mathrm{con}} = 10^{-4}$ and $\beta=0.2$ in the remaining experiments.

The control objective is to track a reference output over the prediction 
horizon $T_f$. We initially consider a standard quadratic cost
\begin{align}
    f_1(u_f) &= \|u_f\|_{\mathbf{R}}^2, \\
    f_2(y_f) &= \|y_f - y_{r}\|_{\mathbf{Q}}^2,
\end{align}
where $y_{r}$ is the stacked reference output over the prediction horizon, 
and $\mathbf{R} = I_{T_f} \otimes R$, $\mathbf{Q} = I_{T_f} \otimes Q$, 
with $R = 0.05 I_m$ and $Q = I_p$. 

To evaluate closed-loop performance, we use the cumulative average cost over the total simulation duration
$T_{\text{run}}=50$
\begin{equation}
\label{eq:test-cost}
    J_{\text{test}} = \frac{1}{T_{\text{run}}} \sum_{k=0}^{T_{\text{run}}-1} 
    \left( \|u_k\|_R^2 + \|y_k - y_{r,k}\|_Q^2 \right).
\end{equation}

We first test the constraint satisfaction performance with respect to changing $\varepsilon_{\mathrm{con}}$ and $\beta$ values. To obtain an informative 
evaluation of constraint handling, we set $y_{r,k}=0$ and impose the output 
constraint $y_k \in [0,2]$, while keeping the input constraint $u_k \in [-2,2]$. The innovation terms $e(t)$ are sampled from a zero-mean Gaussian distribution with covariance $\Sigma_e = 0.012\mathbf{I}_p$ which corresponds to offline output data $\{y_k^d\}_{k=1}^{T}$ with Signal to Noise Ratio (SNR)  of around 10dB.   We consider $\varepsilon_{\mathrm{con}}$ values in $\{10^{-5}, 10^{-4}, 10^{-3}, 10^{-2}, 10^{-1}, 1\}$ and $\beta$ values in $\{0.1, 0.2, 0.5, 0.7, 0.9\}$. The corresponding average constraint violations and the cost performances \eqref{eq:test-cost} of the DR-DDPC method over 50 Monte Carlo simulations are given in Figure~\ref{fig:const-test}. The results show that the empirical violation rate remains within the prescribed risk level $\beta$ for all scenarios. A clear trade-off is observed regarding the ambiguity radius $\varepsilon_{\mathrm{con}}$, where larger values yield stricter constraint satisfaction but result in higher performance costs. Conversely, increasing  $\beta$ allows for lower costs at the expense of more frequent violations.  For comparison, the average violation rates for SPC and Reg-DeePC are $20.96\%$ and $17.80\%$ with average performance costs of $0.2615$ and $0.3276$, respectively. 

\begin{figure}[htbp]
     \centering
     \begin{subfigure}[b]{0.35\textwidth}
         \centering
         \includegraphics[width=\textwidth]{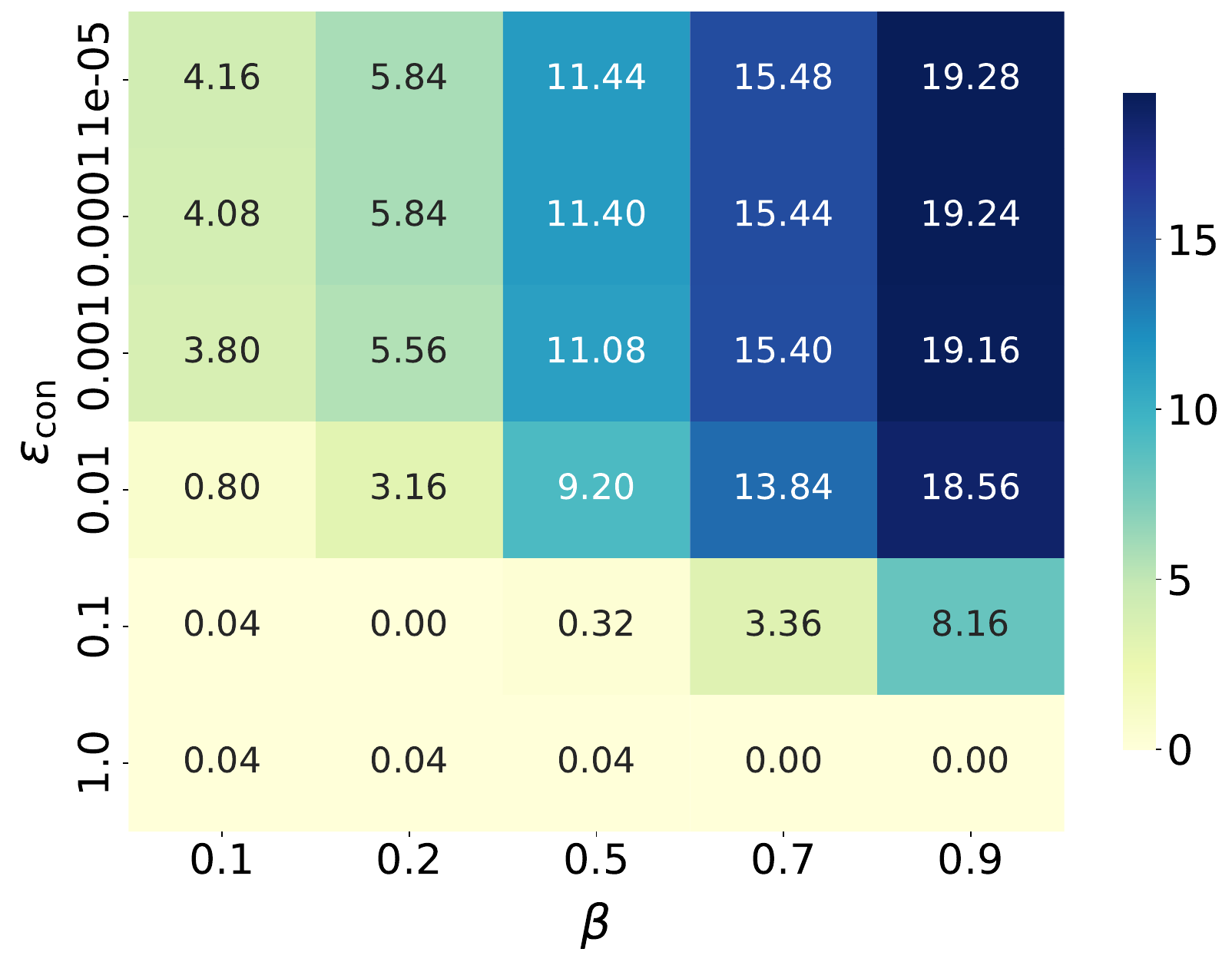}
         \caption{Mean Violation Rate (\%)}
         \label{fig:violation_rate}
     \end{subfigure} 
     \begin{subfigure}[b]{0.35\textwidth}
         \centering
         \includegraphics[width=\textwidth]{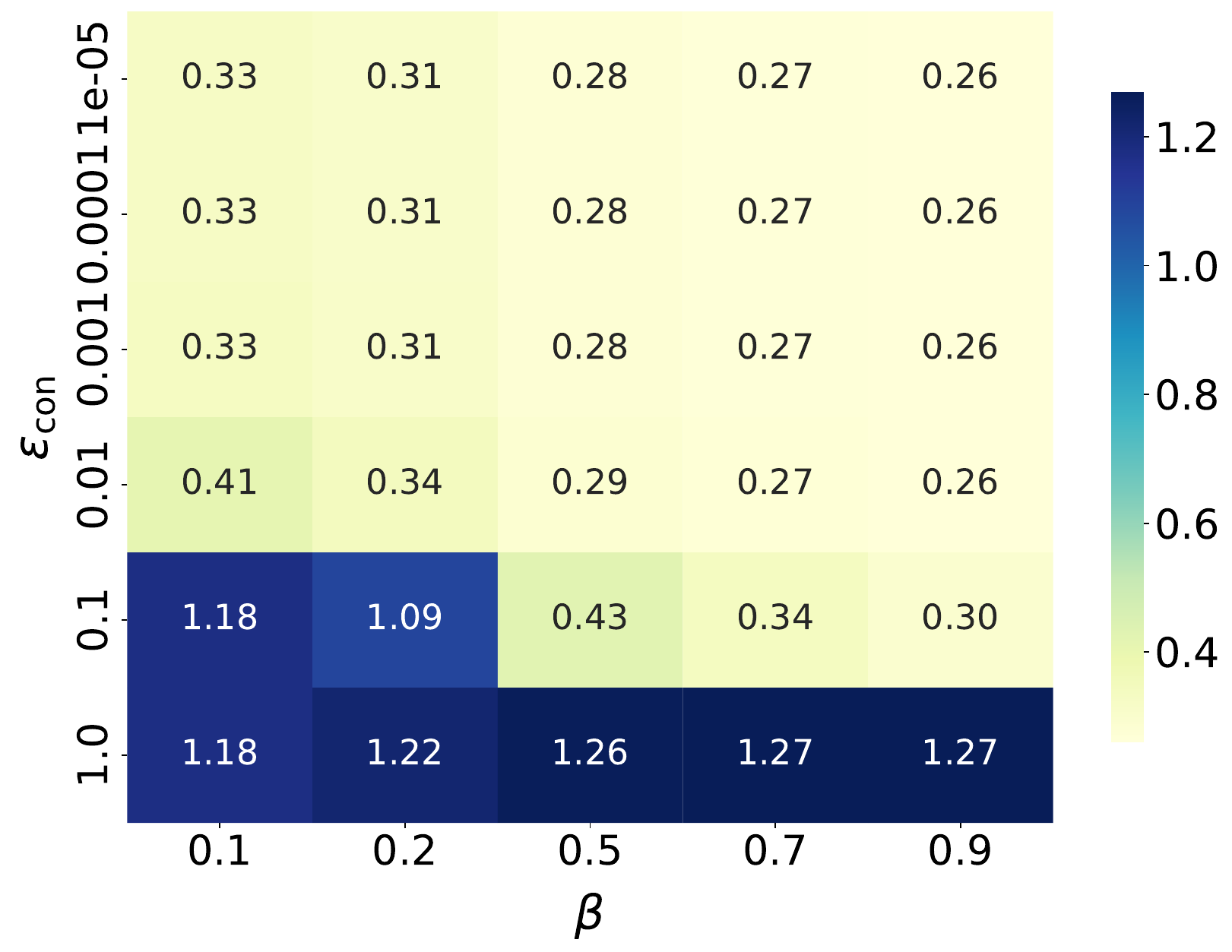}
         \caption{Mean Performance Cost}
         \label{fig:perf_cost}
     \end{subfigure}
     \caption{Parameter sweep results for $\varepsilon_{\mathrm{con}}$ vs $\beta$ using the DR-DDPC controller \eqref{eq:tractable_deepc}.}
     \label{fig:const-test}
\end{figure}

Next, we compare tracking performance across varying noise levels using a 
sinusoidal reference $y_{r,k} = \sin\!\left(\frac{2\pi k}{T_{\text{run}}}\right)$. 
Under zero-mean Gaussian innovations (Table~\ref{tab:gaussian_combined}), Reg-DeePC 
is outperformed by the others, while SPC and DR-DDPC perform nearly identically. 
This is expected as, for quadratic costs, the DR-DDPC objective reduces to the cost 
of the mean scenario, which coincides with the nominal SPC prediction in expectation 
when residuals have zero mean. When the innovation mean is shifted to $0.05$ 
(Table~\ref{tab:gaussian_combined}), a performance gap emerges between SPC and 
DR-DDPC, with DR-DDPC outperforming SPC, as it can account for the bias in the 
predictor residuals.

Finally, we repeat the zero-mean Gaussian experiments while varying the tracking cost function $f_2(y_f)$ to evaluate performance in cases where the scenario spread, rather than just the mean, influences the cost in \eqref{eq:tractable_deepc}. We test an $\ell_1$ cost $f_2(y_f) = \|y_f - y_{r}\|_1$, and an asymmetric linear cost $f_2(y_f) = 2\|(y_f - y_{r})_+\|_1 + \|(y_f - y_{r})_-\|_1$ where $(\cdot)_+$ and $(\cdot)_-$ denote the positive and negative parts to penalize overshooting more heavily. In each case $J_{\text{test}}$ is modified by replacing the quadratic output term in \eqref{eq:test-cost} with the respective cost. As shown in Figure~\ref{fig:noise_sweep_costs}, a performance gap between DR-DDPC and SPC emerges in both cases, confirming that the advantage of DR-DDPC grows when the cost is sensitive to the full scenario distribution rather than just its mean.

\begin{table}[t]
\centering
\caption{Mean $\pm$ Std of $J_{\text{test}}$ for three covariance levels of Gaussian
innovation terms: zero-mean (top) and mean $\mu = 0.05$ (bottom).}
\label{tab:gaussian_combined}
\resizebox{\columnwidth}{!}{%
\begin{tabular}{llccc}
\toprule
& & \multicolumn{3}{c}{$\Sigma_e$} \\
\cmidrule(lr){3-5}
& Method & $0.012$ & $0.0012$ & $1.2\times10^{-5}$ \\
\midrule
\multicolumn{5}{l}{\textit{(a) Zero mean}} \\[2pt]
\multirow{3}{*}{\rotatebox[origin=c]{90}{Cost}}
& SPC       & $0.3062 \pm 0.2382$ & $0.2674 \pm 0.2287$ & $0.2613 \pm 0.2270$ \\
& DR-DDPC   & $0.3065 \pm 0.2381$ & $0.2674 \pm 0.2287$ & $0.2612 \pm 0.2270$ \\
& Reg-DeePC & $0.4740 \pm 0.3026$ & $0.3093 \pm 0.2287$ & $0.2877 \pm 0.2250$ \\
\midrule
\multicolumn{5}{l}{\textit{(b) Mean $\mu = 0.05$}} \\[2pt]
\multirow{3}{*}{\rotatebox[origin=c]{90}{Cost}}
& SPC       & $0.2627 \pm 0.1901$ & $0.2246 \pm 0.1799$ & $0.2124 \pm 0.1772$ \\
& DR-DDPC   & $0.2558 \pm 0.1889$ & $0.2196 \pm 0.1795$ & $0.2119 \pm 0.1771$ \\
& Reg-DeePC & $0.5141 \pm 0.2378$ & $0.3211 \pm 0.1811$ & $0.2530 \pm 0.1771$ \\
\bottomrule
\end{tabular}
}
\end{table}

\begin{figure}[htbp]
    \centering
    \begin{subfigure}[b]{0.23\textwidth}
        \centering
        \includegraphics[width=\textwidth]{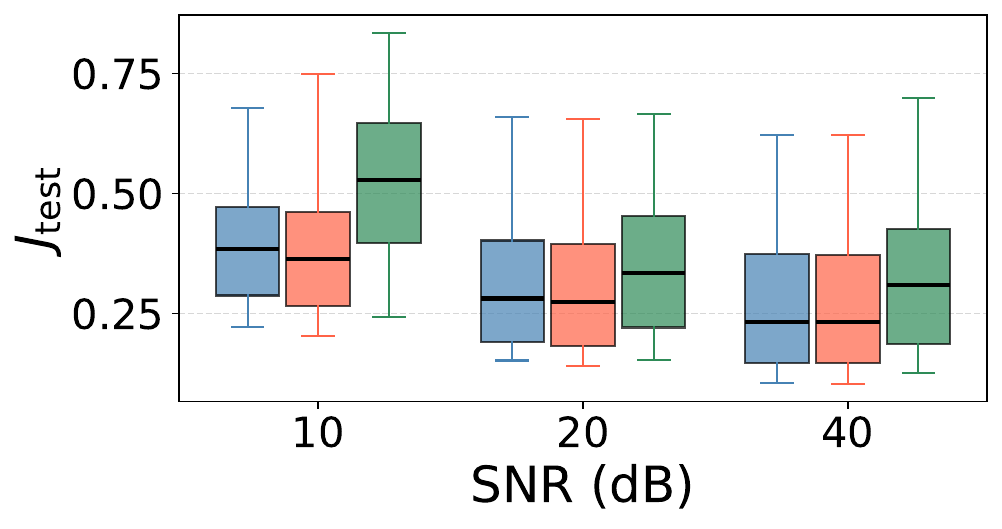}
        \caption{$\ell_1$ cost}
        \label{fig:noise_sweep_l1}
    \end{subfigure}
    \begin{subfigure}[b]{0.23\textwidth}
        \centering
        \includegraphics[width=\textwidth]{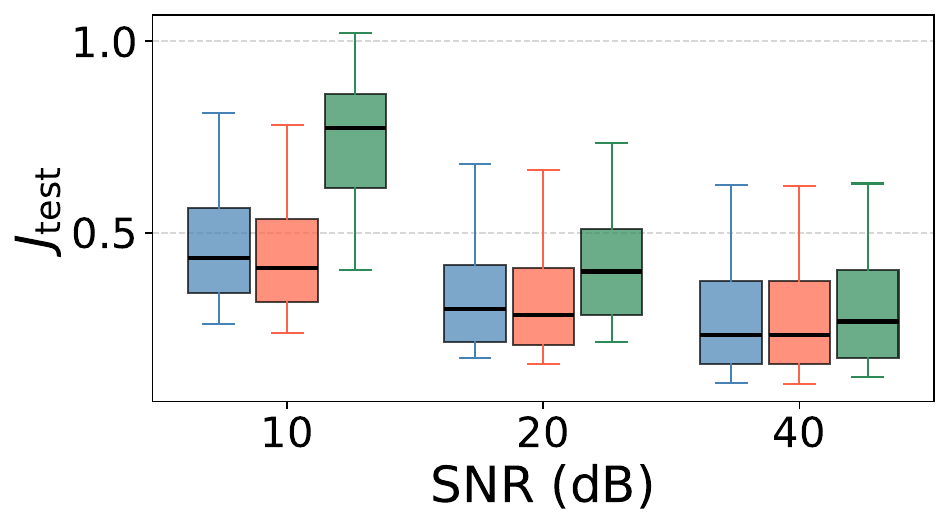}
        \caption{Asymmetric cost}
        \label{fig:noise_sweep_asymmetric}
    \end{subfigure}
    \begin{subfigure}[b]{0.33\textwidth}
        \centering
        \includegraphics[width=\textwidth]{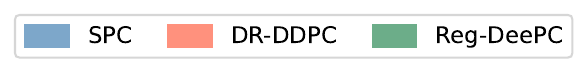}
    \end{subfigure}
    \caption{Cost performance $J_{\text{test}}$ across covariance levels for different cost functions under zero-mean Gaussian innovation terms.}
    \label{fig:noise_sweep_costs}
\end{figure}

\section{Conclusion}
\label{sec:Conc}
In this work, we proposed a distributionally robust data-driven predictive control framework for stochastic LTI systems with unknown dynamics and disturbance distributions. We presented two equivalent formulations, one based on the SPC predictor and one in direct data-driven form, and established finite-sample guarantees on the expected cost and output constraint satisfaction under the true disturbance distribution. Numerical simulations validated the framework against existing methods, showing that DR-DDPC and SPC perform similarly under zero-mean Gaussian innovation terms with quadratic cost, while DR-DDPC outperforms SPC under nonzero-mean innovations and non-quadratic cost functions. Future work includes establishing stability and recursive feasibility guarantees for the closed-loop system.

\appendix

\section{Proof of Theorem~\ref{thm:tractable_spc}}
\label{app:th2}
\begin{proof}
    Applying Th.~1 of \cite{Gao2016DistributionallyRS} and Lemma~A.2 
    of \cite{9488221} to the convex Lipschitz function $f_2$ with 
    constant $L_{\mathrm{obj}}$ directly gives
    \begin{align}
    &f_1(u_f) + \sup_{Q \in \mathcal{B}_\varepsilon(\hat{\mathbb{P}}_{y_f|m_f})} 
    \mathbb{E}_Q[f_2(y_f)] \nonumber \\
    &\le f_1(u_f) + L_{\mathrm{obj}}\varepsilon + 
    \frac{1}{N}\sum_{i=1}^{N} f_2(\hat{K}m_f + \hat{\xi}_f^{(i)}),
    \end{align}
    with equality when $\mathcal{Y} = \mathbb{R}^{pT_f}$. For the 
    CVaR constraint, expanding the definition of CVaR and applying 
    the max-min inequality gives
    \begin{align}
    &\sup_{Q \in \mathcal{B}_\varepsilon(\hat{\mathbb{P}}_{y_f|m_f})} 
    \mathrm{CVaR}^{1-\beta}_Q(h(y_f)) \nonumber \\
    &\le \inf_{\tau \in \mathbb{R}} \left\{ \tau + 
    \sup_{Q \in \mathcal{B}_\varepsilon(\hat{\mathbb{P}}_{y_f|m_f})}
    \frac{1}{\beta}\mathbb{E}_Q\left[(h(y_f)-\tau)_+\right] \right\}.
    \end{align}
    Since $(h(y_f)-\tau)_+$ is convex and Lipschitz continuous in 
    $y_f$ with constant $L_{\mathrm{con}}$, applying Th.~1 of 
    \cite{Gao2016DistributionallyRS} and Lemma~A.2 of \cite{9488221} 
    to the inner supremum yields
    \begin{align}
    &\sup_{Q \in \mathcal{B}_\varepsilon(\hat{\mathbb{P}}_{y_f|m_f})} 
    \mathrm{CVaR}^{1-\beta}_Q(h(y_f)) \nonumber \\
    &\le \inf_{\tau \in \mathbb{R}} \left\{ \tau + 
    \frac{L_{\mathrm{con}}\varepsilon}{\beta} + 
    \frac{1}{\beta N}\sum_{i=1}^{N}
    \left(h(\hat{K}m_f + \hat{\xi}_f^{(i)})-\tau\right)_+ \right\}.
    \end{align}
    Introducing slack variables $s_i \ge 
    (h(\hat{K}m_f+\hat{\xi}_f^{(i)})-\tau)_+ \ge 0$, 
    the infimum is at most zero whenever 
    \eqref{eq:tr_cvar}--\eqref{eq:tr_slack_pos} are satisfied, 
    completing the proof.
\end{proof}

\section{ Proof of Lemma~\ref{lem:wasserstein_concentration}}
\label{app:wass}
We first introduce the dependence concepts used in the proof.
\begin{definition}[$m$-dependence \cite{Gu2023CentralLT}]
\label{def:mdep}
A sequence $\{Z_i\}_{i \geq 1}$ is \emph{$m$-dependent} if for every $n$ 
and $j \geq m+1$, the block $(Z_{n+m+1}, \ldots, Z_{n+j})$ is independent 
of $(Z_1, \ldots, Z_n)$.
\end{definition}

\begin{definition}[$\rho$-mixing {\cite{Fournier2013OnTR}}]
\label{def:rhomix}
A stationary sequence $(X_n)_{n\ge1}$ with distribution $\mu$ is 
$\rho$-mixing if there exists $\rho:\mathbb{N}\to\mathbb{R}_+$ with 
$\rho_n \to 0$ such that for all $f,g\in L^2(\mu)$ and $i,j\ge1$,
$\mathrm{Cov}(f(X_i),g(X_j)) \le \rho_{|i-j|}
\sqrt{\mathrm{Var}(f(X_i))\,\mathrm{Var}(g(X_j))}.$
\end{definition}

\begin{proof}[Proof of Lemma~\ref{lem:wasserstein_concentration}]
We verify the conditions of Theorem~14 in \cite{Fournier2013OnTR}, namely 
stationarity and $\rho$-mixing with $\sum_{n\ge 0}\rho_n < \infty$, for 
the sequence $\{\bar{y}_f^{(i)}(m_f)\}_{i\ge1}$ obtained by hypothetically 
extending the offline trajectory.

\textit{Stationarity.} Since $\bar{y}_f^{(i)}(m_f) = Km_f + 
\bar{\xi}_f^{(i)}$ and $Km_f$ is constant, stationarity of 
$\{\bar{y}_f^{(i)}(m_f)\}_{i\ge1}$ is equivalent to stationarity of 
$\{\bar{\xi}_f^{(i)}\}_{i\ge1}$. Each $\bar{\xi}_f^{(i)}$ is a measurable 
function of the noise window $\{w_k,v_k\}_{k=i}^{i+L-1}$ with $L=T_p+T_f$. 
Since $\{w_k\}$ and $\{v_k\}$ are i.i.d., shifted windows have identical 
joint distributions, hence $\{\bar{\xi}_f^{(i)}\}_{i\ge1}$ is stationary.

\textit{$L$-dependence and $\rho$-mixing.} For $|i-j|\ge L$, the noise 
windows of $\bar{\xi}_f^{(i)}$ and $\bar{\xi}_f^{(j)}$ do not overlap and 
are therefore independent by the i.i.d.\ assumption. Hence 
$\{\bar{y}_f^{(i)}(m_f)\}_{i\ge1}$ is $L$-dependent, which implies 
$\rho$-mixing with $\rho_n = 0$ for all $n \ge L$, so 
$\sum_{n\ge 0}\rho_n < \infty$.

All conditions of Theorem~14 in \cite{Fournier2013OnTR} are now satisfied, 
and together with the moment condition $\mathbb{E}_{\mathbb{P}_{y_f|m_f}}
[\|y_f\|_r^q] < \infty$ for some $q>r$, applying that theorem gives
\begin{equation*}
\mathbb{E}_{\mathbb{P}}\!\left[d_{W_r}^r\!\left(\bar{\mathbb{P}}_{y_f|m_f},\, 
\mathbb{P}_{y_f|m_f}\right)\right] \le \gamma(N).
\end{equation*}
Applying Markov's inequality then yields
\begin{equation*}
\mathbb{P}\!\left(d_{W_r}\!\left(\bar{\mathbb{P}}_{y_f|m_f}, 
\mathbb{P}_{y_f|m_f}\right)\ge\kappa\right) \le 
\frac{\gamma(N)}{\kappa^r}. \qed
\end{equation*}
\end{proof}

\bibliographystyle{plain}
\bibliography{ref.bib}

\begin{thebibliography}{10}

\bibitem{9109670}
Julian Berberich, Johannes Köhler, Matthias~A. Müller, and Frank Allgöwer.
\newblock Data-driven model predictive control with stability and robustness guarantees.
\newblock {\em IEEE Transactions on Automatic Control}, 66(4):1702--1717, 2021.

\bibitem{BRESCHI2023110961}
Valentina Breschi, Alessandro Chiuso, and Simone Formentin.
\newblock Data-driven predictive control in a stochastic setting: a unified framework.
\newblock {\em Automatica}, 152:110961, 2023.

\bibitem{9488221}
Jeremy Coulson, John Lygeros, and Florian D\"orfler.
\newblock Distributionally robust chance constrained data-enabled predictive control.
\newblock {\em IEEE Transactions on Automatic Control}, 67(7):3289--3304, 2022.

\bibitem{8795639}
Jeremy Coulson, John Lygeros, and Florian Dörfler.
\newblock Data-enabled predictive control: In the shallows of the deepc.
\newblock In {\em 2019 18th European Control Conference (ECC)}, pages 307--312, 2019.

\bibitem{9705109}
Florian Dörfler, Jeremy Coulson, and Ivan Markovsky.
\newblock Bridging direct and indirect data-driven control formulations via regularizations and relaxations.
\newblock {\em IEEE Transactions on Automatic Control}, 68(2):883--897, 2023.

\bibitem{Esfahani2015DatadrivenDR}
Peyman~Mohajerin Esfahani and Daniel Kuhn.
\newblock Data-driven distributionally robust optimization using the wasserstein metric: performance guarantees and tractable reformulations.
\newblock {\em Mathematical Programming}, 171:115 -- 166, 2015.

\bibitem{FAVOREEL19994004}
Wouter Favoreel, Bart~De Moor, and Michel Gevers.
\newblock Spc: Subspace predictive control.
\newblock {\em IFAC Proceedings Volumes}, 32(2):4004--4009, 1999.
\newblock 14th IFAC World Congress 1999, Beijing, Chia, 5-9 July.

\bibitem{Fournier2013OnTR}
Nicolas Fournier and Arnaud Guillin.
\newblock On the rate of convergence in wasserstein distance of the empirical measure.
\newblock {\em Probability Theory and Related Fields}, 162:707 -- 738, 2013.

\bibitem{Gao2016DistributionallyRS}
Rui Gao and Anton~J. Kleywegt.
\newblock Distributionally robust stochastic optimization with wasserstein distance.
\newblock {\em Math. Oper. Res.}, 48:603--655, 2016.

\bibitem{Gu2023CentralLT}
Wangyun Gu and Lixu Zhang.
\newblock Central limit theorem for m-dependent random variables under sub-linear expectations.
\newblock 2023.

\bibitem{berberichConst}
Christian Klöppelt, Julian Berberich, Frank Allgöwer, and Matthias~A. Müller.
\newblock A novel constraint-tightening approach for robust data-driven predictive control.
\newblock {\em International Journal of Robust and Nonlinear Control}, 35(7):2566--2587, 2025.

\bibitem{9992469}
Francesco Micheli, Tyler Summers, and John Lygeros.
\newblock Data-driven distributionally robust mpc for systems with uncertain dynamics.
\newblock In {\em 2022 IEEE 61st Conference on Decision and Control (CDC)}, pages 4788--4793, 2022.

\bibitem{nemirovski}
Arkadi Nemirovski and Alexander Shapiro.
\newblock Convex approximations of chance constrained programs.
\newblock {\em SIAM Journal on Optimization}, 17(4):969--996, 2007.

\bibitem{9999463}
Guanru Pan, Ruchuan Ou, and Timm Faulwasser.
\newblock On a stochastic fundamental lemma and its use for data-driven optimal control.
\newblock {\em IEEE Transactions on Automatic Control}, 68(10):5922--5937, 2023.

\bibitem{10896586}
Malika Sader, Yibo Wang, Dexian Huang, Chao Shang, and Biao Huang.
\newblock Causality-informed data-driven predictive control.
\newblock {\em IEEE Transactions on Control Systems Technology}, 33(5):1921--1928, 2025.

\bibitem{11192654}
Johannes Teutsch, Sebastian Kerz, Dirk Wollherr, and Marion Leibold.
\newblock Sampling-based stochastic data-driven predictive control under data uncertainty.
\newblock {\em IEEE Transactions on Automatic Control}, 71(3):1924--1931, 2026.

\bibitem{willems2005note}
Jan~C Willems, Paolo Rapisarda, Ivan Markovsky, and Bart~LM De~Moor.
\newblock A note on persistency of excitation.
\newblock {\em Systems \& Control Letters}, 54(4):325--329, 2005.

\end{thebibliography}
\end{document}